\documentclass{llncs}

\usepackage{amsmath}
\usepackage{graphicx}
\usepackage{caption}
\usepackage{amssymb}
\usepackage{amsfonts}
\usepackage{authblk}

\usepackage{lipsum}
\usepackage{xargs}
\usepackage[pdftex,dvipsnames]{xcolor}

\title{On the Computational Complexity of MapReduce}
\date{}

\author{Benjamin Fish\inst{1}\thanks{Contact author. Address: 851 S.\ Morgan St.\ Chicago, IL 60607. Phone:  (312)-996-3041.} \and
Jeremy Kun\inst{1} \and
\'Ad\'am D.\ Lelkes\inst{1} \and
Lev Reyzin\inst{1} \and
Gy\"orgy Tur\'an\inst{1,2}}

\institute{
Department of Mathematics, Statistics, and Computer Science,\\ University of Illinois at Chicago
\\ \email{\{bfish3,jkun2,alelke2,lreyzin,gyt\}@uic.edu}
\and
MTA-SZTE Research Group on Artificial Intelligence, Szeged
}

\newcommand{\N}{\mathbb{N}}
\newcommand{\mrc}{\textup{MRC}}
\newcommand{\bsp}{\textup{BSP}}
\newcommand{\SPACE}{\textup{SPACE}}
\newcommand{\TIME}{\textup{TIME}}
\newcommand{\TISP}{\textup{TISP}}
\renewcommand{\P}{\textup{P}}
\renewcommand{\L}{\textup{SPACE}(\log(n))}
\newcommand{\NP}{\textup{NP}}
\newcommand{\PSPACE}{\textup{PSPACE}}

\begin{document}
\maketitle

\begin{abstract} 

In this paper we study the MapReduce Class (MRC) defined by Karloff~et~al.,
which is a formal complexity-theoretic model of MapReduce. We show that
constant-round MRC computations can decide regular languages and simulate
sublogarithmic space-bounded Turing machines. In addition, we prove hierarchy
theorems for MRC under certain complexity-theoretic assumptions. These theorems
show that sufficiently increasing the number of rounds or the amount of time
per processor strictly increases the computational power of MRC. Our work lays
the foundation for further analysis relating MapReduce to established
complexity classes. Our results also hold for Valiant's BSP model of parallel
computation and the MPC model of Beame~et~al. 

\end{abstract}
\thispagestyle{empty}
\newpage
\setcounter{page}{1}
\section{Introduction}

MapReduce is a programming model originally developed to separate algorithm
design from the engineering challenges of massively distributed computing. A
programmer can separately implement a ``map'' function and a ``reduce''
function that satisfy certain constraints, and the underlying MapReduce
technology handles all the communication, load balancing, fault tolerance, and
scaling. MapReduce frameworks and their variants have been successfully
deployed in industry by Google~\cite{DeanG08}, Yahoo!~\cite{ShvachkoKRC10}, and
many others.

MapReduce offers a unique and novel model of parallel computation because it
alternates parallel and sequential steps, and imposes sharp constraints on
communication and random access to the data. This distinguishes MapReduce from
classical theoretical models of parallel computation and this, along with its
popoularity in industry, is a strong motivation to study the theoretical power
of MapReduce. From a theoretical standpoint we ask how MapReduce relates to
established complexity classes. From a practical standpoint we ask which
problems can be efficiently modeled using MapReduce and which cannot.

In 2010 Karloff~et~al.~\cite{Karloff10} initiated a principled theoretical
study of MapReduce, providing the definition of the complexity class MRC and
comparing it with the classical PRAM models of parallel computing. But to our
knowledge, since this initial paper, almost all of the work on MapReduce has
focused on algorithmic issues.

Complexity theory studies the classes of problems defined by resource bounds on
different models of computation in which they are solved. A central goal of
complexity theory is to understand the relationships between different models,
i.e. to see if the problems solvable with bounded resources on one
computational model can be solved with a related resource bound on a different
model.  In this paper we prove a result that establishes a connection between
MapReduce and space-bounded computation on classical Turing machines.  Another
traditional question asked by complexity theory is whether increasing the
resource bound on a certain computational resource strictly increases the set
of solvable problems. Such so-called hierarchy theorems exist for time and
space on deterministic and non-deterministic Turing machines, among other
settings.  In this paper we prove conditional hierarchy theorems for MapReduce
rounds and time.

First we lay a more precise theoretical foundation for studying MapReduce
computations (Section~\ref{sec:definition}).  In particular, we observe that
Karloff et al.'s definitions are non-uniform, allowing the complexity class to
contain undecidable languages.  We reformulate the definition of
\cite{Karloff10} to make a uniform model and to more finely track the
parameters involved (Section~\ref{subsec:nonuniformity}).  In addition, we
point out that our results hold for other important models of parallel
computations, including Valiant's Bulk-Synchronous Processing (BSP)
model~\cite{Valiant90} and the Massively Parallel Communication (MPC) model of
Beame~et~al~\cite{BeameKS13}.  (Section~\ref{subsec:bspmodel}).  We then prove
two main theorems: $\SPACE(o(\log n))$ has constant-round MapReduce
computations (Section~\ref{sec:spacebound}) and, conditioned on a version of
the Exponential Time Hypothesis, there are strict hierarchies within MRC.  In
particular, sufficiently increasing time or number of rounds increases the
power of MRC (Section~\ref{sec:hierarchy}).

Our sub-logarithmic space result is achieved by a direct simulation, using a
two-round protocol that localizes state-to-state transitions to the section of
the input being simulated, combining the sections in the second round. It is a
major open problem whether undirected graph connectivity (a canonical
logarithmic-space problem) has a constant-round MapReduce algorithm, and our
result is the most general that can be proven without a breakthrough on graph
connectivity.  Our hierarchy theorem involves proving a conditional time
hierarchy within linear space achieved by a padding argument, along with
proving a time-and-space upper and lower bounds on simulating MRC machines
within $\P$. To the best of our knowledge our hierarchy theorem is the first of
its kind. We conclude with a discussion and open questions raised by our work
(Section~\ref{sec:openproblems}).

\section{Background and Previous Work}

\subsection{MapReduce}

The MapReduce protocol can be roughly described as follows. The input data is
given as a list of key-value pairs, and over a series of rounds two things
happen per round: a ``mapper'' is applied to each key-value pair independently
(in parallel), and then for each distinct key a ``reducer'' is applied to all
corresponding values for a group of keys. The canonical example is counting
word frequencies with a two-round MapReduce protocol. The inputs are (index,
word) pairs, the first mapper maps $(k,v) \mapsto (v,k)$, and the first reducer
computes the sum of the word frequencies for the given key. In the second round
the mapper sends all data to a single processor via $(k, n_k) \mapsto (1, (k,
n_k))$, and the second processor formats the output appropriately.

One of the primary challenges in MapReduce is data locality. MapReduce was
designed for processing massive data sets, so MapReduce programs require that
every reducer only has access to a substantially sublinear portion of the
input, and the strict modularization prohibits reducers from communicating
within a round. All communication happens indirectly through mappers, which are
limited in power by the independence requirement. Finally, it's understood in
practice that a critical quantity to optimize for is the number of
rounds~\cite{Karloff10}, so algorithms which cannot avoid a large number of
rounds are considered inefficient and unsuitable for MapReduce.

There are a number of MapReduce-like models in the literature, including the
MRC model of Karloff~et~al.~\cite{Karloff10}, the ``mud'' algorithms of
Feldman~et~al.~\cite{FeldmanMSSS10}, Valiant's BSP model~\cite{Valiant90}, the
MPC model of Beame~et~al.~\cite{BeameKS13}, and extensions or generalizations
of these, e.g.~\cite{GoodrichSZ11}. The MRC class of Karloff~et~al. is the
closest to existing MapReduce computations, and is also among the most
restrictive in terms of how it handles communication and tracks the
computational power of individual processors. In their influential paper
\cite{Karloff10}, Karloff~et~al. display the algorithmic power of MRC, and
prove that MapReduce algorithms can simulate CREW PRAMs which use subquadratic
total memory and processors. It is worth noting that the work of Karloff~et~al.
did not include comparisons to the standard (non-parallel) complexity classes,
which is the aim of the present work.

Since \cite{Karloff10}, there has been extensive work in developing efficient
algorithms in MapReduce-like frameworks. For example,
Kumar~et~al.~\cite{KMVV13} analyze a sampling technique allowing them to
translate sequential greedy algorithms into $\log$-round MapReduce algorithms
with a small loss of quality.  Farahat~et~al.~\cite{FEGK13} investigate the
potential for sparsifying distributed data using random projections. Kamara and
Raykova~\cite{KR13} develop a homomorphic encryption scheme for MapReduce. And
much work has been done on graph problems such as connectivity, matchings,
sorting, and searching~\cite{GoodrichSZ11}. Chu~et~al.~\cite{ChuKLYBNO06}
demonstrate the potential to express any statistical-query learning algorithm
in MapReduce.  Finally, Sarma~et~al.~\cite{Sarma13} explore the relationship
between communication costs and the degree to which a computation is parallel
in one-round MapReduce problems. Many of these papers pose general upper and
lower bounds on MapReduce computations as an open problem, and to the best of
our knowledge our results are the first to do so with classical complexity
classes.

The study of MapReduce has resulted in a wealth of new and novel algorithms,
many of which run faster than their counterparts in classical PRAM models. As
such, a more detailed study of the theoretical power of MapReduce is warranted.
Our paper contributes to this by establishing a more precise definition of the
MapReduce complexity class, proving that it contains sublogarithmic
deterministic space, and showing the existence of certain kinds of hierarchies.

\subsection{Complexity}

From a complexity-theory viewpoint, MapReduce is unique in that it combines
bounds on time, space and communication. Each of these bounds would be very
weak on its own: the total time available to processors is polynomial; the
total space and communication are slightly less than quadratic.  In particular,
even though arranging the communication between processors is one of the most
difficult parts of designing MapReduce algorithms, classical results from
communication complexity do not apply since the total communication available
is more than linear. These innocent-looking bounds lead to serious restrictions
when combined, as demonstrated by the fact that it is unknown whether
constant-round MRC machines can decide graph connectivity (the best known
result achieves a logarithmic number of rounds with high
probability~\cite{Karloff10}), although it is solvable using only logarithmic
space on a deterministic Turing machine.

We relate the MRC model to more classical complexity classes by studying simultaneous time-space bounds. $\TISP(T(n),S(n))$ are the problems that
can be decided by a Turing machine which on inputs of length $n$ takes at most
$O(T(n))$ time and uses at most $O(S(n))$ space. Note that in general it is
believed that $\TISP(T(n), S(n)) \neq \TIME(T(n)) \cap \SPACE(S(n))$.  The
complexity class $\TISP$ is studied in the context of time-space tradeoffs
(see, for example,~\cite{Fortnow00,Williams08}). Unfortunately much less is
known about $\TISP$ than about $\TIME$ or $\SPACE$; for example there is no
known time hierarchy theorem for fixed space. The existence of such a hierarchy
is mentioned as an open problem in the monograph of Wagner and
Wechsung~\cite{WagnerW86}.

To prove the results about $\TISP$ that imply the existence of a hierarchy in
MRC, we use the Exponential Time Hypothesis (ETH) introduced by Impagliazzo,
Paturi, and Zane~\cite{ImpagliazzoP99,ImpagliazzoPZ01}, which conjectures that
3-SAT is not in $\TIME(2^{cn})$ for some $c>0$. This hypothesis and its strong
version have been used to prove conditional lower bounds for specific hard
problems like vertex cover, and for algorithms in the context of fixed
parameter tractability (see, e.g., the survey of Lokshtanov, Marx and
Saurabh~\cite{LokshtanovMS11}). The first open problem mentioned
in~\cite{LokshtanovMS11} is to relate ETH to some other known complexity
theoretic hypotheses.

We show in Lemma~\ref{lemma:eth} that ETH implies directly a time-space
trade-off statement involving time-space complexity classes. This statement is
not a well-known complexity theoretic hypothesis, although it is related to the
existence of a time hierarchy with a fixed space bound. In fact, as detailed in
Section~\ref{sec:hierarchy}, a hypothesis weaker than ETH is sufficient for the
lemma. The relative strengths of ETH, the weaker hypothesis, and the statement
of the lemma seem to be unknown.

\section{Models} \label{sec:definition}

In this section we introduce the model we will use in this paper, a uniform
version of Karloff's MapReduce Class (MRC), and contrast MRC to other models of
parallel computation, such as Valiant's Bulk-Synchronous Parallel (BSP) model,
for which our results also hold.

\subsection{MapReduce and MRC}

The central piece of data in MRC is the key-value pair, which we denote by a
pair of strings $\langle k, v \rangle$, where $k$ is the key and $v$ is the
value. An input to an MRC machine is a list of key-value pairs $\langle k_i,
v_i \rangle_{i=1}^N$ with a total size of $n = \sum_{i=1}^N |k_i| + |v_i|$.
The definitions in this subsection are adapted from~\cite{Karloff10}.

\begin{definition}
A \emph{mapper} $\mu$ is a Turing machine\footnote{The definitions
of~\cite{Karloff10} were for RAMs. However, because we wish to relate MapReduce
to classical complexity classes, we reformulate the definitions here in terms
of Turing machines.} which accepts as input a single key-value pair $\langle k,
v \rangle$ and produces a list of key-value pairs $\langle k_1', v_1' \rangle,
\dots, \langle k_s', v_s' \rangle$. \end{definition}

\begin{definition}
A \emph{reducer} $\rho$ is a Turing machine which accepts as input a key $k$
and a list of values $\langle v_1 , \dots, v_m \rangle$, and produces as output
the same key and a new list of values $\langle v_1', \dots, v_M' \rangle$.
\end{definition}

\begin{definition}
For a decision problem, an input string $x \in \{ 0,1 \}^*$ to an MRC machine
is the list of pairs $\langle i, x_i \rangle_{i=1}^n$ describing the index
and value of each bit. We will denote by $\langle x \rangle$ the list $\langle
i, x_i \rangle$.
\end{definition}

An MRC machine operates in rounds. In each round, a set of mappers running in
parallel first process all the key-value pairs. Then the pairs are partitioned
(by a mechanism called ``shuffle and sort'' that is not considered part of the
runtime of an MRC machine) so that each reducer only receives key-value pairs
for a single key. Then the reducers process their data in parallel, and the
results are merged to form the list of key-value pairs for the next round. More
formally:

\begin{definition}
An \emph{$R$-round MRC machine} is an alternating list of mappers and reducers
$M = (\mu_1, \rho_1, \dots, \mu_R, \rho_R)$. The execution of the machine is as
follows.
For each $r = 1, \dots, R$:
\begin{enumerate}
  \item Let $U_{r-1}$ be the list of key-value pairs generated by round $r-1$
(or the input pairs when $r=1$). Apply $\mu_r$ to each key-value pair of
$U_{r-1}$ to get the multiset $V_r = \bigcup_{\langle k,v \rangle \in U_{r-1}}
\mu_r(k, v).$

  \item Shuffle-and-sort groups the values by key. Call each of the pieces $V_{k,r} =$
$\{ k, (v_{k,1}, \dots, v_{k,s_k})\}.$

  \item Assign a different copy of reducer $\rho_r$ to each $V_{k,r}$ (run in
parallel) and set $U_r = \bigcup_{k}\rho_r(V_{k,r})$.
\end{enumerate}
\end{definition}

The output is the final set of key-value pairs. For decision problems, we
define $M$ to accept $\left \langle x \right \rangle$ if in the final round
$U_R = \emptyset$. Equivalently we may give each reducer a special accept state
and say the machine accepts if at any time any reducer enters the accept state.
We say $M$ \emph{decides} a language $L$ if it accepts $\langle x \rangle$ if
and only if $x \in L$.

The central caveat that makes MRC an interesting class is that the reducers
have space constraints that are sublinear in the size of the input string. In
other words, no sequential computation may happen that has random access to the
entire input. Thinking of the reducers as processors, cooperation between
reducers is obtained not by message passing or shared memory, but rather across
rounds in which there is a global communication step.

In the MRC model we use in this paper, we require that every mapper and reducer
arise as separate runs of the same Turing machine $M$. Our Turing machine $M(m,
r, n, y)$ will accept as input the current round number $r$, a bit $m$ denoting
whether to run the $r$-th map or reduce function, the total size of the input
$n$, and the corresponding input $y$.  Equivalently, we can imagine a list of
mappers and reducers in each round $\mu_1, \rho_1, \mu_2, \rho_2, \dots$, where
the descriptions of the $\mu_i, \rho_i$ are computable in
polynomial time in $|i|$. 

\begin{definition}[Uniform Deterministic MRC]

A language $L$ is said to be in $\mrc[f(n),g(n)]$ if there is a constant $0 < c
< 1$, an $O(n^c)$-space and $O(g(n))$-time Turing machine $M(m, r, n, y)$, and
an $R = O(f(n))$, such that for all $x \in \{ 0,1 \}^n$, the following holds.

\begin{enumerate}
\item Letting $\mu_r = M(1, r, n, -), \rho_r = M(0, r, n, -)$, the MRC machine
$M_R = (\mu_1, \rho_1, \dots, \mu_R, \rho_R)$ accepts $x$ if and only if $x \in
L$.

\item Each $\mu_r$ outputs $O(n^c)$ distinct keys.
\end{enumerate}

\end{definition}

This definition closely hews to practical MapReduce computations:  $f(n)$
represents the number of times global communication has to be performed, $g(n)$
represents the  time each processor gets, and sublinear space bounds in terms of $n =
|x|$ ensure that the size of the data on each processor is smaller than the
full input.

\begin{remark}
By $M(1, r, n, -)$, we mean that the tape of $M$ is initialized by the string
$\langle 1, r, n \rangle$. In particular, if the number of rounds is $2^{\omega(n^c)}$ for all constants $c<1$, the space constraints would prohibit it from storing the round number, which prohibits any MRC algorithm from having this many rounds.
\end{remark}

\begin{remark}\label{remark:timebound}
Note that a polynomial time Turing machine with sufficient time can trivially
simulate a uniform MRC machine. All that is required is for the machine to
perform the key grouping manually, and run the MRC machine as a subroutine. As
such, $\mrc[\text{poly}(n), \text{poly}(n)] \subseteq P$. We give a more precise
computation of the amount of overhead required in the proof of
Lemma~\ref{lem:tisprelation}.
\end{remark}


\begin{definition}
Define by $\mrc^i$ the union of uniform MRC classes
\[
   \mrc^i = \bigcup_{k\in\N}\mrc[\log^i(n),n^k].
\]
\end{definition}

So in particular $\mrc^0=\bigcup_{k\in\N}\text{MRC}[1,n^k]$. 

\subsection{Nonuniformity}  \label{subsec:nonuniformity}

A complexity class is generally called uniform if the descriptions of the
machines solving problems in it do not depend on the input length. Classical
complexity classes defined by Turing machines with resource bounds, such as
$\P$, $\NP$, and $\L$, are uniform. On the other hand, circuit complexity
classes are naturally nonuniform since a fixed Boolean circuit can only accept
inputs of a single length. There is ambiguity about the uniformity of MRC as
defined in~\cite{Karloff10}. Since we wish to relate the MRC model to classical
complexity classes such as $\P$ and $\L$, making sure that the model is uniform
is crucial. Indeed, innocuous-seeming changes to the definitions above
introduce nonuniformity (and in particular this is true of the original MRC
definition in~\cite{Karloff10}). In the appendix we show that the nonuniform
MRC model defined in~\cite{Karloff10} allows MRC machines to solve undecidable
problems in a logarithmic number of rounds, including the halting problem. We
introduced our uniform version of MRC above to rule out such pathological
behavior.

\subsection{Other Models of Parallel Computation} \label{subsec:bspmodel}

Several other models of parallel computation have been introduced, including
the BSP model of Valiant~\cite{Valiant90} and the MPC model of
Beame~et.~al.~\cite{BeameKS13}.  The main difference between BSP and MapReduce
is that in the BSP models the key-value pairs and the shuffling steps needed to
redistribute them are replaced with point-to-point messages.  Similarly
to~\cite{Karloff10}, in Valiant's paper~\cite{Valiant90} there is also
ambiguity about the uniformity of the model. In this paper, when we refer to
BSP we mean a uniform deterministic version of the model. We give the exact
definition in the appendix.

Goodrich~et al.~\cite{GoodrichSZ11} and Pace~\cite{Pace12} showed that
MapReduce computations can be simulated in the BSP model and vice versa, with
only a constant blow-up in the computational resources needed.  This implies
that our theorems about MapReduce automatically apply to BSP.

Similarly, the MPC model uses point-to-point messages and Beame et.~al.'s
paper~\cite{BeameKS13} does not discuss the uniformity of the model.  The main
distinguishing charateristic of the MPC model is that it introduces the number
of processors $p$ as an explicit paramter.  Setting $p=O(n^c)$, our results
will also hold in this model.

There are other variants of these models, including the model that
Andoni~et.~al.~\cite{AndoniNOY14} uses, which follows the MPC model but also
introduces the additional constraint that total space used across each round
must be no more than $O(n)$.  It is straightforward to check that the proofs of
our results never use more than $O(n)$ space, implying that our results hold
even under this more restrictive model.

\section{Space Complexity Classes in $\mrc^0$} \label{sec:spacebound}

In this section we prove that small space classes are contained in
constant-round MRC.  Again, the results in this section also hold for other
similar models of parallel computation, including the BSP model and the MPC model.
First, we prove that the class REGULAR of regular languages
is in $\mrc^0$. It is well known that $\textup{SPACE}(O(1)) = \textup{REGULAR}$
\cite{Shepherdson59}, and so this result can be viewed as a warm-up to the
theorem that $\textup{SPACE}(o(\log n)) \subseteq \mrc^0$. Indeed, both proofs
share the same flavor, which we sketch before proceeding to the details.

We wish to show that any given DFA can be simulated by an $\mrc^0$ machine.
The simulation works as follows:
in the first round each parallel processor receives a contiguous portion of the
input string and constructs a state transition function using the data of the
globally known DFA. Though only the processor with the beginning of the string
knows the true state of the machine during its portion of the input, all
processors can still compute the \emph{entire} table of state-to-state
transitions for the given portion of input. In the second round, one processor
collects the transition tables and chains together the computations, and this
step requires only the first bit of input and the list of tables.

We can count up the space and time requirements to prove the following theorem.

\begin{theorem}\label{thm:reg}
   $\textup{REGULAR} \subsetneq \mrc^0$
\end{theorem}

\begin{proof}
Let $L$ be a regular language and $D$ a deterministic finite automaton
recognizing $L$. Define the first mapper so
that the $j^\textup{th}$ processor has the bits from $j\sqrt{n}$ to
$(j+1)\sqrt{n}$. This means we have $K = O(\sqrt{n})$ processors in the first
round. Because the description of $D$ is independent of the size of the input
string, we also assume each processor has access to the relevant set of states
$S$ and the transition function $t: S \times \left \{ 0,1 \right \} \to S$.

We now define $\rho_1$. Fix a processor $j$ and call its portion of the input
$y$. The processor constructs a table $T_j$ of size at most $|S|^2 = O(1)$ by
simulating $D$ on $y$ starting from all possible states and recording the state
at the end of the simulation. It then passes $T_j$ and the first bit of $y$ to
the single processor in the second round.

In the second round the sole processor has $K$ tables $T_j$ and the first bit
$x_1$ of the input string $x$ (among others but these are ignored). Treating
$T_j$ as a function, this processor computes $q = T_K(\dots T_2(T_1(x_1)))$ and
accepts if and only if $q$ is an accepting state. This requires $O(\sqrt{n})$
space and time and proves containment. To show this is strict, inspect the
prototypical problem of deciding whether the majority of bits in the input are
1's.
\end{proof}

\begin{remark} \label{rmk:tradeoff}
While the definition of $\mrc^0$ inclues languages with time complexity
$O(n^k)$ for all $k \geq 0$, our Theorem~\ref{thm:reg} is more efficient than
the definition implies: we show that regular languages can be computed in
$\mrc^0$ in time and space $O(\sqrt{n})$, with the option of a tradeoff between
time $n^\varepsilon$ and space $n^{1-\varepsilon}$.  
\end{remark}

One specific application of this result is that for any given regular
expression, a two-round MapReduce computation can decide if a string matches
that regular expression, even if the string is so long that any one machine can
only store $n^\epsilon$ bits of it.

We now move on to prove $\SPACE(o(\log n))\subseteq \mrc^0$. It is worth noting
that this is a strictly stronger statement than Theorem~\ref{thm:reg}. That
is, $\textup{REGULAR}=\SPACE(O(1))\subsetneq \SPACE(o(\log n))$. Several
non-trivial examples of languages that witness the strictness of this
containment are given in~\cite{szepietowski1994turing}.

The proof is very similar to the proof of Theorem~\ref{thm:reg}:  Instead of
the processors computing the entire table of state-to-state transitions of a
DFA, the processors now compute the entire table of all transitions possible
among the configurations of the work tape of a Turing machine that uses $o(\log
n)$ space. 

\begin{theorem} \label{thm:sublogspace}
$\textup{SPACE}(o(\log n))\subseteq \mrc^0$.
\end{theorem}

\begin{proof}
Let $L$ be a language in $\textup{SPACE}(o(\log n))$ and $T$ a Turing machine
recognizing $L$ in polynomial time and $o(\log(n))$ space, with a read/write
work tape $W$. Define the first mapper so that the $j^\textup{th}$ processor
has the bits from $j\sqrt{n}$ to $(j+1)\sqrt{n}$. Let $\mathcal{C}$ be the set
of all possible configurations of $W$ and let $S$ be the states of $T$. Since
the size of $S$ is independent of the input, we can assume that each processor
has the transition function of $T$ stored on it.

Now we define $\rho_1$ as follows: Each processor $j$ constructs the graph of a
function $T_j:\mathcal{C} \times \{L,R\} \times S \rightarrow \mathcal{C}
\times \{L,R\} \times S$, which simulates $T$ when the read head starts on
either the left or right side of the $j$th $\sqrt{n}$ bits of the input and $W$
is in some configuration from $\mathcal{C}$. It outputs whether the read head
leaves the $y$ portion of the read tape on the left side, the right side, or
else accepts or rejects. To compute the graph of $T_j$, processor $j$
simulates $T$ using its transition function, which takes polynomial time.

Next we show that the graph of $T_j$ can be stored on processor $j$ by showing
it can be stored in $O(\sqrt{n})$ space. Since $W$ is by assumption size
$o(\log n)$, each entry of the table is $o(\log n)$, so there are $2^{o(\log
n)}$ possible configurations for the tape symbols. There are also $o(\log n)$
possible positions for the read/write head, and a constant number of states $T$
could be in. Hence $|\mathcal{C}| = 2^{o(\log n)} o(\log n) = o(n^{1/3})$.
Then processor $j$ can store the graph of $T_j$ as a table of size
$O(n^{1/3})$.

The second map function $\mu_2$ sends each $T_j$ (there are $\sqrt{n}$ of them)
to a single processor. Each is size $O(n^{1/3})$, and there are $\sqrt{n}$ of
them, so a single processor can store all the tables. Using these tables, the
final reduce function can now simulate $T$ from starting state to either the
accept or reject state by computing $q=T_k^*(\ldots
T_2^*(T_1^*(\emptyset,L,initial)))$ for some $k$, where $\emptyset$ denotes the
initial configuration of $T$, $initial$ is the initial state of $T$, and $q$ is
either in the accept or reject state. Note $T_j^*$ is the modification of $T_j$
such that if $T_j(x)$ outputs $L$, then $T_j^*(x)$ outputs $R$ and vice versa.
This is necessary because if the read head leaves the $j^\textup{th}$
$\sqrt{n}$ bits to the right, it enters the ${j+1}^\textup{th}$ $\sqrt{n}$ bits
from the left, and vice versa. Finally, the reducer accepts if and only if $q$ is in an
accept state.

This algorithm successfully simulates $T$, which decides $L$, and only takes a
constant number of rounds, proving containment. \end{proof}

\section{Hierarchy Theorems}\label{sec:hierarchy}

In this section we prove two main results (Theorems~\ref{thm:roundhierarchy}
and~\ref{thm:timehierarchy}) about hierarchies within MRC relating to increases
in time and rounds.  They imply that allowing MRC machines sufficiently more
time or rounds strictly increases the computing power of the machines.  The
first theorem states that for all $\alpha, \beta$ there are problems $L \not
\in \mrc[n^\alpha, n^\beta]$ which can be decided by \emph{constant time} MRC
machines when given enough extra rounds.

\begin{theorem} \label{thm:roundhierarchy}
Suppose the ETH holds with constant $c$. Then for every $\alpha,\beta\in\mathbb N$
there exists a $\gamma = O(\alpha + \beta)$ such that $$\mrc[n^\gamma,1]
\not\subseteq \mrc[n^\alpha, n^\beta].$$
\end{theorem}

The second theorem is analogous for time, and says that there are problems $L
\not \in \mrc[n^\alpha, n^\beta]$ that can be decided by a \emph{one round} MRC
machine given enough extra time.

\begin{theorem} \label{thm:timehierarchy}
Suppose the ETH holds with constant $c$. Then for every $\alpha,\beta\in\mathbb N$
there exists a $\gamma = O(\alpha + \beta)$ such that $$\mrc[1, n^\gamma]
\not\subseteq \mrc[n^\alpha, n^\beta].$$
\end{theorem}

As both of these theorems depend on the ETH, we first prove a complexity-theoretic
lemma that uses the ETH to give a time-hierarchy within linear space $\TISP$.
Recall that $\TISP$ is the complexity class defined by simultaneous time and 
space bounds.
The lemma can also be described as a time-space tradeoff.
For some $b > a$ we prove
the existence of a language that can be decided by a Turing machine with
simultaneous $O(n^b)$ time and linear space, but cannot be decided by a Turing
machine in time $O(n^a)$ even without any space restrictions. It is widely
believed such languages exist for \emph{exponential} time classes (for example,
TQBF, the language of true quantified Boolean formulas, is a linear space
language which is PSPACE-complete). We ask whether such tradeoffs can be
extended to polynomial time classes, and this lemma shows that indeed this is
the case.

\begin{lemma}\label{lemma:eth} Suppose that the ETH holds with constant $c$. Then
for any positive integer $a$ there exists a positive integer $b>a$ such that
$$ \TISP(n^b, n) \nsubseteq \TIME(n^a).$$ \end{lemma}

\begin{proof}
By the ETH, 3-SAT $\in \TISP(2^n, n)\setminus \TIME(2^{cn})$. Let
$b:=\lceil\frac{a}{c}\rceil+2$, $\delta:=\frac12(\frac1b+\frac{c}{a})$. Pad
3-SAT with $2^{\delta n}$ zeros and call this language $L$, i.e. let $L:=\{x
0^{2^{\delta|x|}} \mid x \in \textup{3-SAT}\}$. Let $N:=n+2^{\delta n}$. Then
$L\in \TISP(N^b, N)$ since $N^b > 2^n$. On the other hand, assume for
contradiction that $L \in \TIME(N^a)$. Then, since $N^a < 2^{cn}$, it follows
that $\textup{3-SAT} \in \TIME(2^{cn})$, contradicting the ETH.
\end{proof}

There are a few interesting complexity-theoretic remarks about the above proof.
First, the starting language does not need to be 3-SAT, as the only assumption
we needed was its hypothesized time lower bound.  We could relax the assumption
to the hypothesis that there exists a $c>0$ such that TQBF,
the $\PSPACE$-complete language of true quantified Boolean formulas,
requires $2^{cn}$ time, or further still to the following complexity hypothesis.

\begin{conjecture} \label{conj:weaketh}
There exist $c',c$ satisfying $0 < c' < c < 1$ such that
$\TISP(2^n, 2^{c'n}) \setminus \TIME(2^{cn})\ne\emptyset$.
\end{conjecture}

Second, since $\TISP(n^a, n) \subseteq \TIME(n^a)$, this conditionally proves the
existence of a hierarchy within $\TISP(\textup{poly}(n), n)$. We note that
finding time hierarchies in fixed-space complexity classes was posed as an open
question by~\cite{WagnerW86}, and so removing the hypothesis or replacing it
with a weaker one is an interesting open problem.

Using this lemma we can prove Theorems~\ref{thm:roundhierarchy}
and~\ref{thm:timehierarchy}. The proof of Theorem~\ref{thm:roundhierarchy}
depends on the following lemma.

\begin{lemma} \label{lem:tisprelation}
For every $\alpha, \beta \in \mathbb N$ the following holds: $$ \TISP(n^\alpha,
n)  \subseteq \mrc[n^\alpha, 1] \ \subseteq \mrc[n^\alpha, n^\beta] \subseteq
\TISP(n^{\alpha+\beta+2}, n^2).$$
\end{lemma}

\begin{proof} The first inequality follows from a simulation argument similar
to the proof of Theorem~\ref{thm:sublogspace}. The $\mrc$ machine will simulate
the $\TISP(n^\alpha,n)$ machine by making one step per round, with the tape
(including the possible extra space needed on the work tape) distributed among
the processors.  The position of the tape is passed between the processors from
round to round.  It takes constant time to simulate one step of the
$\TISP(n^\alpha,n)$ machine, thus in $n^\alpha$ rounds we can simulate all
steps.  Also, since the machine uses only linear space, the simulation can be
done with $O(\sqrt{n})$ processors using $O(\sqrt{n})$ space each.  The second
inequality is trivial.

The third inequality is proven as follows. Let $T(n) = n^{\alpha+\beta+2}$.
We first show that any language in $\mrc[n^\alpha,n^\beta]$ can be simulated in
time $O(T(n))$, i.e. $\mrc[n^\alpha,n^\beta] \subseteq \TIME(T(n))$. The $r$-th
round is simulated by applying $\mu_r$ to each key-value pair in sequence,
shuffle-and-sorting the new key-value pairs, and then applying $\rho_r$ to each
appropriate group of key-value pairs sequentially. Indeed, $M(m,r,n,-)$ can be
simulated naturally by keeping track of $m$ and $r$, and adding $n$ to the tape
at the beginning of the simulation. Each application of $\mu_r$ takes
$O(n^\beta)$ time, for a total of $O(n^{\beta+1})$ time. Since each mapper
outputs no more than $O(n^c)$ keys, and each mapper and reducer is in
$\SPACE(O(n^c))$, there are no more than $O(n^2)$ keys to sort. Then
shuffle-and-sorting takes $O(n^2 \log n)$ time, and the applications of
$\rho_r$ also take $O(n^{\beta+1})$ time. So a round takes $O(n^{\beta+1} +
n^2 \log n)$ time. Note that keeping track of $m$,$r$, and $n$ takes no more
than the above time. So over $O(n^\alpha)$ rounds, the simulation takes
$O(n^{\alpha+\beta+1}+n^{\alpha+2} \log(n))=O(T(n))$ time. \end{proof}

Now we prove Theorem~\ref{thm:roundhierarchy}.

\begin{proof}
By Lemma~\ref{lemma:eth}, there is a language $L$ in $\TISP(n^\gamma,
n)\setminus \TIME(n^{\alpha+\beta+2})$ for some $\gamma$. By
Lemma~\ref{lem:tisprelation}, $L\in \mrc[n^\gamma, 1]$. On the other hand,
because $L\not\in  \TIME(n^{\alpha+\beta+2})$ and $\mrc[n^\alpha, n^\beta]\subseteq
\TIME(n^{\alpha+\beta+2})$, we can conclude that $L\not\in
\mrc[n^\alpha, n^\beta]$.
\end{proof}

Next, we prove Theorem~\ref{thm:timehierarchy} using a padding argument.

\begin{proof}
Let $T(n) = n^{\alpha + \beta + 2}$ as in Lemma~\ref{lem:tisprelation}. By
Lemma~\ref{lemma:eth}, there is a $\gamma$ such that $\TISP(n^\gamma,
n)\setminus\TIME(T(n^{2}))$ is nonempty.  Let $L$ be a language from this set.
Pad $L$ with $n^{2}$ zeros, and call this new language $L'$, i.e. let $L' = \{x
0^{|x|^{2}} \mid x \in  L\}$. Let $N = n + n^{2}$. There is an $\mrc[1,
N^\gamma]$ algorithm to decide $L'$: the first mapper discards all the
key-value pairs except those in the first $n$, and sends all remaining pairs to
a single reducer. The space consumed by all pairs is $O(n) = O(\sqrt{N})$. This
reducer decides $L$, which is possible since $L \in \TISP(n^\gamma, n)$. We now
claim $L'$ is not in $\mrc[N^\alpha, N^\beta]$. If it were, then $L'$ would be
in $\TIME(T(N))$. A Turing machine that decides $L'$ in $T(N)$ time can be
modified to decide $L$ in $T(N)$ time: pad the input string with $n^{2}$ ones
and use the decider for $L'$. This shows $L$ is in $\TIME(T(n^{2}))$, a
contradiction.  \end{proof}

We conclude by noting explicitly that
Theorems~\ref{thm:roundhierarchy},~\ref{thm:timehierarchy} give proper
hierarchies within MRC, and that proving certain stronger hierarchies imply the
separation of L and P.

\begin{corollary} \label{cor:mrchierarchy}
Suppose the ETH. For every $\alpha, \beta$ there exist $\mu>\alpha$ and $\nu>\beta$
such that $$\mrc[n^\alpha, n^\beta] \subsetneq \mrc[n^\mu, n^\beta]$$ and
$$\mrc[n^\alpha, n^\beta] \subsetneq \mrc[n^\alpha, n^\nu].$$
\end{corollary}

\begin{proof}
By Theorem \ref{thm:timehierarchy}, there is some $\mu > \alpha$ such that
$\mrc[n^\mu, 1] \not\subseteq \mrc[n^\alpha, n^\beta]$.  It is immediate that
$\mrc[n^\alpha, n^\beta] \subseteq \mrc[n^\mu, n^\beta]$ and also that
$\mrc[n^\mu,1]\subseteq\mrc[n^\mu,n^\beta]$.  So $\mrc[n^\alpha,n^\beta]\neq
\mrc[n^\mu, n^\beta]$.  The proof of the second claim is similar.
\end{proof}

\begin{corollary} \label{cor:lvsp}
If $\mrc[\textup{poly}(n), 1] \subsetneq \mrc[\textup{poly}(n),
\textup{poly}(n)]$, then it follows that $\L \neq \P$.
\end{corollary}

\begin{proof}
\begin{align*}	\L &\subseteq \TISP(\textup{poly}(n), \log n) \subseteq
\TISP(\textup{poly}(n), n) \subseteq \mrc[\textup{poly}(n), 1] \\ &\subseteq
\mrc[\textup{poly}(n), \textup{poly}(n)] \subseteq \P.\end{align*}

The first containment is well known, the third follows from 
Lemma~\ref{lem:tisprelation}, and the rest are trivial.
\end{proof}

Corollary~\ref{cor:lvsp} is interesting because if any of the containments in
the proof are shown to be proper, then $\L \neq \P$. Moreover, if we provide
MRC with a polynomial number of rounds, Corollary~\ref{cor:lvsp} says that
determining whether time provides substantially more power is at least as hard
as separating $\L$ from P. On the other hand, it does not rule out the possibility
that $\mrc[\textup{poly}(n), \textup{poly}(n)] = \P$, or even that
$\mrc[\textup{poly}(n), 1] = \P$.

\section{Discussion and Open Problems}\label{sec:openproblems}

In this paper we established the first general connections between MapReduce
and classical complexity classes, and showed the conditional existence of a
hierarchy within MapReduce. Our results also apply to variants of MapReduce,
most notably Valiant's BSP model.  


Our work suggests some natural open problems. How does MapReduce relate to
other complexity classes, such as the circuit class uniform $\textup{AC}^0$?
Can one improve the bounds from Corollary~\ref{cor:mrchierarchy} or remove the
dependence on Hypothesis~\ref{conj:weaketh}? Does Lemma~\ref{lemma:eth} imply
Hypothesis~\ref{conj:weaketh}? Can one give explicit hierarchies for space or
time alone, e.g. $\mrc[n^\alpha, \text{poly}(n)] \subsetneq \mrc[n^\mu,
\text{poly}(n)]$?

We also ask whether $\mrc[\text{poly}(n), \text{poly}(n)] = \P$. In other
words, if a problem has an efficient solution, does it have one with using data
locality? A negative answer implies $\L\neq\P$ which is a major open problem in
complexity theory, and a positive answer would likely provide new and valuable
algorithmic insights. Finally, while we have focused on the relationship
between rounds and time, there are also implicit parameters for the amount of
(sublinear) space per processor, and the (sublinear) number of processors per
round. A natural complexity question is to ask what the relationship
between all four parameters are.

\section*{Acknowledgements} We thank Howard Karloff and Benjamin Moseley for
helpful discussions.

\bibliographystyle{plain}
\bibliography{paper}

\section*{Appendix}

\appendix
\section{Nonuniform MRC}
In this section we show that the original MRC definition of \cite{Karloff10}
allows MRC machines to decide undecidable languages.
This definition required a polylogarithmic number
of rounds, and also allowed completely different MapReduce machines for
different input sizes.
For simplicity's sake, we will allow a linear number of rounds,
and use our notation $\mrc[f(n), g(n)]$ to denote an MRC machine
that operates in $O(f(n))$ rounds and each processor gets $O(g(n))$ time
per round.
In particular, we show that
nonuniform $\mrc[n, \sqrt{n}]$ accepts all unary languages, i.e. languages of
the form $L \subseteq \{1^n \mid n \in \N\}$.

\begin{lemma}\label{lemma:unary}
Let $L$ be a unary language. Then $L$ is in nonuniform $\mrc[n, \sqrt{n}]$.
\end{lemma}

\begin{proof}
We define the mappers and reducers as follows. Let $\mu_1$ distribute the input
as contiguous blocks of $\sqrt{n}$ bits, $\rho_1$ compute the length of its
input, $\mu_2$ send the counts to a single processor, and $\rho_2$ add up the
counts, i.e.\ find $n=|x|$ where $x$ is the input. Now the input data is
reduced to one key-value pair $\langle \star, n \rangle$. Then let $\rho_i$ for
$i \ge 3$ be the reducer that on input $\langle \star, i-3 \rangle$ accepts if
and only if $1^{i-3} \in L$ and otherwise outputs the input. Let $\mu_i$ for $i
\ge 3$ send the input to a single processor. Then $\rho_{n+3}$ will accept iff
$x$ is in $L$. Note that $\rho_1, \rho_2$ take $O(\sqrt{n})$ time, and all
other mappers and reducers take $O(1)$ time. All mappers and reducers are also
in $\SPACE(\sqrt{n})$. \end{proof}

In particular, Lemma~\ref{lemma:unary} implies that nonuniform $\mrc[n,
\sqrt{n}]$ contains the unary version of the halting problem. A more careful
analysis shows all unary languages are even in $\mrc[\log n, \sqrt{n}]$, by
having $\rho_{i+3}$ check $2^i$ strings for membership in $L$.

\section{Uniform BSP}
We define the BSP model of Valiant~\cite{Valiant90} similarly to MRC, where
essentially key-value pairs are replaced with point-to-point messages.

A BSP machine with $p$ processors is a list $(M_1, \dots, M_p)$ of $p$ Turing
machines which on any input, output a list $((j_1, y_1), (j_2, y_2), \dots,
(j_m, y_m))$ of messages to be sent to other processors in the next round.
Specifically, message $y_k$ is sent to prcessor $j_k$. A BSP machine operates
in rounds as follows. In the first round the input is partitioned into
equal-sized pieces $x_{1,0}, \dots, x_{p,0}$ and distributed arbitrarily to the
processors. Then for rounds $r=1, \dots, R$, 

\begin{enumerate}
   \item Each processor $i$ takes $x_{i,r}$ as input and computes some number
$s_i$ of messages $M_i(x_{i,r}) = \{(j_{i,k}, y_{i,k}) : k = 1, \dots, s_i\}$.
   \item Set $x_{i,r+1}$ to be the set of all messages sent to $i$ (as with
MRC's shuffle-and-sort, this is not considered part of processor $i$'s
runtime). 
\end{enumerate} 

We say the machine \emph{accepts} a string $x$ if any machine accepts at any
point before round $R$ finishes. We now define uniform deterministic BSP
analogously to MRC.

\begin{definition}[Uniform Deterministic BSP]

A language $L$ is said to be in $\bsp[f(n),g(n)]$ if there is a constant $0 < c
< 1$, an $O(n^c)$-space and $O(g(n))$-time Turing machine $M(p, y)$, and an $R
= O(f(n))$, such that for all $x \in \{ 0,1 \}^n$, the following holds: letting
$M_i = M(i, -)$, the BSP machine $M = (M_1, M_2, \dots, M_{n^c})$ accepts $x$
in $R$ rounds if and only if $x \in L$.

\end{definition}

\begin{remark}
As with MRC, we count the size and number of each message as part of the space
bound of the machine generating/receiving the messages. Differing slightly from
Valiant, we do not provide persistent memory for each processor. Instead we
assume that on processor $i$, any memory cell not containing a message will
form a message whose destination is $i$. This is without loss of generality
since we are not concerned with the cost of sending individual messages.
\end{remark}

\end{document}